%% file: manuscript_version7.tex
\documentclass[aps, reprint, nofootinbib, longbibliography, superscriptaddress]{revtex4-1}

\usepackage{amsmath, amssymb, amsthm}
\usepackage{bm,bbm}% bold math

\usepackage{graphicx}
\usepackage{hyperref}
\usepackage[margin=2cm]{geometry}
\usepackage{xcolor}
\usepackage{qcircuit}
\usepackage{multirow}
\usepackage[normalem]{ulem}
\usepackage{enumitem}
\usepackage{ulem}
\usepackage{booktabs}
\theoremstyle{plain}
\AtBeginDocument{
\heavyrulewidth=.08em
\lightrulewidth=.05em
\cmidrulewidth=.03em
\belowrulesep=.65ex
\belowbottomsep=0pt
\aboverulesep=.4ex
\abovetopsep=0pt
\cmidrulesep=\doublerulesep
\cmidrulekern=.5em
\defaultaddspace=.5em
}

\input{./commands.tex}

\renewcommand{\vec}[1]{\ensuremath{\boldsymbol{#1}}}

\newcommand{\ket}[1]{\ensuremath{\left| #1 \right \rangle}}
\newcommand{\bra}[1]{\ensuremath{\left \langle #1 \right |}}

\newcommand{\bx}{\boldsymbol{x}}
\newcommand{\bn}{\boldsymbol{n}}

\newcommand{\btheta}{\boldsymbol \theta}

\newtheorem*{theorem*}{Theorem}

\begin{document}
\title{The effect of data encoding on the expressive power of variational quantum machine learning models}
\author{Maria Schuld}
\affiliation{Xanadu, Toronto, ON, M5G 2C8, Canada}

\author{Ryan Sweke}
\affiliation{Dahlem Center for Complex Quantum Systems, Freie Universit\"{a}t Berlin, 14195 Berlin, Germany}

\author{Johannes Jakob Meyer}
\affiliation{Dahlem Center for Complex Quantum Systems, Freie Universit\"{a}t Berlin, 14195 Berlin, Germany}

\date{\today}

\begin{abstract}
Quantum computers can be used for supervised learning by treating parametrised quantum circuits as models that map data inputs to predictions. While a lot of work has been done to investigate practical implications of this approach, many important theoretical properties of these models remain unknown. Here we investigate how the strategy with which data is encoded into the model influences the expressive power of parametrised quantum circuits as function approximators. We show that one can naturally write a quantum model as a partial Fourier series in the data, where the accessible frequencies are determined by the nature of the data encoding gates in the circuit. By repeating simple data encoding gates multiple times, quantum models can access increasingly rich frequency spectra. We show that there exist quantum models which can realise all possible sets of Fourier coefficients, and therefore, if the accessible frequency spectrum is asymptotically rich enough, such models are universal function approximators.
\end{abstract}

\maketitle

A popular approach to quantum machine learning uses trainable quantum circuits as machine learning models similar to neural networks. Quantum gates -- the building blocks of quantum circuits -- are used to encode data inputs $\bx = (x_1,\dots,x_N)$ as well as trainable weights $\boldsymbol \theta = (\theta_1,\dots, \theta_M)$. The circuit is measured multiple times to estimate the expectation of some observable, and the result is interpreted as a prediction. The overall computation implements a  ``quantum model function'' $f_{\btheta}(\bx)$, a machine learning model that is based on quantum computing. This approach is known by different names such as \textit{variational circuits} \cite{mcclean2016theory, romero2019variational},  \textit{quantum circuit learning} \cite{mitarai2018quantum}, \textit{quantum neural networks} \cite{farhi2018classification, mcclean2018barren}, or \textit{parametrised quantum circuits} \cite{benedetti2019parameterized}. 
 
A lot of work has been done to understand the practical details of this approach, 
leading to useful training strategies \cite{mitarai2018quantum, ostaszewski2019quantum, stokes2020quantum}
and ways to emulate and extend classical machine learning methods \cite{verdon2019quantum_a,verdon2019quantum_b,romero2019variational, cong2019quantum, liu2018differentiable}. A growing body of literature, motivated by the dilemma of investigating the performance of quantum machine learning when only small-scale experiments are physically possible, tries to understand the potential power of variational circuits from a theoretical perspective \cite{harrow2019low, mcclean2018barren, ciliberto2020fast, cerezo2020costfunctiondependent}.
Still, only little is known about the actual \textit{function classes} that quantum circuits give rise to. Can quantum models express any function in the input $\bx$, or are they limited to a specific class of functions? Can this class of ``learnable functions" be characterised in a meaningful way, and can the characterisation be used to guide design choices and potential applications for these quantum models? 

\begin{figure}
    \centering
    \includegraphics[width=0.45\textwidth]{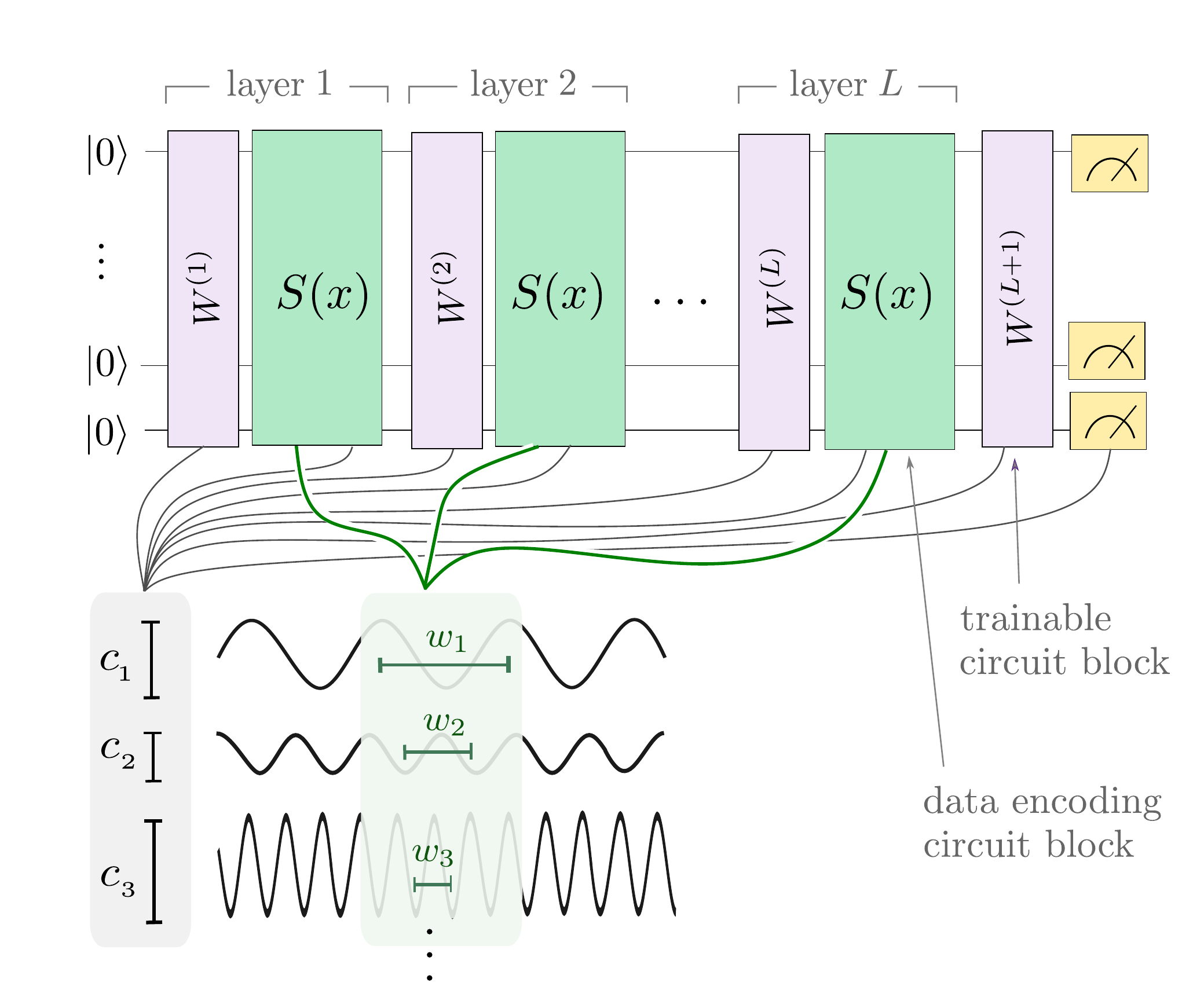}
    \caption{Illustration of the main result of this paper, shown for one-dimensional inputs $x \in \mathbb{R}$: quantum models consisting of layers of trainable circuit blocks $W=W(\btheta)$ and data encoding circuit blocks $S(x)$ can be written as a weighed sum $\sum_{\omega} c_{\omega} e^{i \omega x}$. The data encoding circuit determines the frequencies $\omega$, and the remainder of the circuit architecture determines the coefficients $c_{\omega}$. If the $\omega$ are integer-valued (or integer-valued multiples of a base frequency $\omega_0$), the sum becomes a partial Fourier series, which allows us to systematically study properties of the function class a given quantum model can learn. }
    \label{fig:scheme}
\end{figure}

In this paper we investigate these questions in a framework focused on the role of data encoding.
We consider standard models from the literature that consist of multiple ``circuit layers'', each made up of a \textit{data encoding (circuit) block} and a \textit{trainable (circuit) block}, and assume that input features $x\in \mathbb{R}$ are encoded by gates of the form $e^{ixH}$, where $H$ is an arbitrary Hamiltonian. Our main tool is the natural representation of such quantum models as a Fourier-type sum
\begin{equation}
    f_{\btheta}(\bx) = \sum_{\oomega \in \Omega} c_{\oomega}(\btheta) e^{i \oomega \bx},
    \label{eq:fourier_type}
\end{equation}
where $\oomega \bx$ is the inner product. We show that the \emph{frequency spectrum} $\Omega \subset \mathbb{R}^N$ is solely determined by the eigenvalues of the data-encoding Hamiltonians 
while the design of the entire circuit controls the coefficients $c_{\oomega}$ that a quantum model can realise (see Fig.\ref{fig:scheme}). The representation of quantum models as Fourier-type sums characterises the function families that a given class of quantum models can learn via two interrelated properties. The first property is the frequency spectrum $\Omega$, which determines the functions $e^{i \oomega \bx}$ that the quantum model ``has access to''. The second property is the \textit{expressivity of the coefficients} $\{c_{\oomega} \}$ that a class of quantum models can control, which determines how the accessible functions can be combined.
In many natural settings, the frequencies are integers, $\Omega \subset \mathbb{Z}^N$, and the sum becomes a multi-dimensional \textit{partial Fourier series}
\begin{equation}
    f_{\btheta}(\bx) = \sum_{\bn \in \Omega} c_{\bn}(\btheta) e^{i \bn \bx},
    \label{eq:fourier_series}
\end{equation}
where the $e^{i \bn \bx}$ are orthogonal basis functions. We use the nomenclature \textit{partial} Fourier series to indicate the fact that only a subset of the Fourier coefficients are non-zero. The Fourier series formalism allows us to study quantum models using the rich techniques developed in Fourier analysis. 

First, we consider the popular strategy of encoding an input into single-qubit rotations, and show that repeating the encoding $r$ times either sequentially or in parallel allows the model to access frequency spectra $\Omega$ consisting of $r$ frequencies. This places into a broader context an observation made in Ref.~\cite{ostaszewski2019quantum}, which states that encoding a data feature only once into the angle of a single qubit rotation restricts the function class that quantum models can learn to a simple sine function (or equivalently, a Fourier series with a single frequency). Second, we provide bounds for the maximum number of frequencies and Fourier coefficients a quantum model can control for more general data encoding strategies. Finally, we study the \textit{universality} of quantum models. We show that for sufficiently flexible trainable circuit blocks there exists a quantum model which can realise any possible set of Fourier coefficients. If, asymptotically, the accessible frequency spectrum is rich enough, then such models are universal function approximators. This follows from the fact that Fourier series with arbitrary coefficients can approximate any square integrable function on a given interval~\cite{carleson1966convergence}. 

A few existing studies are related to our work. For example, P{\'e}rez-Salinas et al.~\cite{perez2020data} considered quantum models with sequentially repeated data encodings and conjectured that they are universal function approximators 
under a special kind of classical data pre-processing. Killoran et al.~\cite{killoran2019continuous} have shown that many neural networks can be naturally emulated on a photonic quantum computer, and point out that such quantum models therefore inherit universality. The majority of quantum machine learning papers concerned with questions of expressivity and universality \cite{sim2019expressibility, chen2018universal, du2018expressive, biamonte2019universal}, however, interpret these concepts from a quantum information perspective, which asks whether a circuit can express \textit{any quantum computation}, not any function in the inputs. However, in the context of (supervised) machine learning, quantum universality does not necessarily imply universal function approximation; a quantum circuit able to realise arbitrary unitary evolutions may only be able to express a limited class of functions $f(\bx)$.\footnote{As an extreme example, consider a parametrised quantum circuit that encodes the data into gates acting on qubits which are never entangled with the measured qubits -- in which case $f(\bx)$ is a constant function, and the resulting machine learning model trivial.} From the \textit{function} expressivity view-point, Ref.~\cite{caro2020pseudo} has investigated the pseudo-dimension of a particular class of quantum models, an expressivity metric which allows one to characterize learnability and generalization power of the associated model. Also the essential role of data encoding for quantum machine learning has been emphasised in previous papers. For example, it was remarked that data encoding determines the features that quantum models represent \cite{schuld2019quantum, havlivcek2019supervised}, the decision boundaries they can learn \cite{larose2020robust}, as well as the measurements that optimally distinguish between data classes \cite{lloyd2020quantum}. A central contribution of this paper is to systematically combine the study of data encoding with that of the expressivity of quantum models. 

We present our results as follows: Section \ref{Sec:tool} introduces the basic idea of writing quantum models as partial Fourier series. Section \ref{Sec:expressivity} puts the tool to use and analyses the expressivity of quantum models, which leads to a proof that quantum models are universal in Section \ref{Sec:universality}. Section \ref{Sec:implications} discusses practically relevant implications.

\textit{Note: After publishing the preprint of this article, we were made aware that the connection between Fourier series and quantum machine learning models with repeated data-encoding has already been established in Ref.~\cite{vidal2019input}. While there is significant overlap between this work and ours, we provide a novel universality result, as well as a systematic development of this connection through practically relevant examples.}

\section{Quantum models as partial Fourier series} \label{Sec:tool}

First, we introduce our basic tool: the natural representation of a quantum model as a partial Fourier series. For simplicity, the majority of our presentation will focus on the case of univariate functions with inputs $ x\in\mathbb{R}$, but we generalise this to multivariate functions in Appendix~\ref{app:multivariate_FS}, which is used for the analysis of universality in Section \ref{Sec:universality}.

We define a (univariate) quantum model $f_{\btheta}(x)$ as the expectation value of some observable with respect to a state prepared via a parametrised quantum circuit, i.e.
\begin{equation}
    f_{\boldsymbol \theta}(x) = \bra{0} U^{\dagger}(x,\boldsymbol \theta) M U(x, \boldsymbol \theta) \ket{0},
    \label{eq:model}
\end{equation}
where $|0\rangle$ is some initial state of the quantum computer, $U(x, \btheta)$ is a quantum circuit that depends on the input $x$ and a (possibly empty) set of parameters $\btheta$,
and $M$ is some observable. The prediction of the quantum model at a specific point $x$ is estimated in practice by running the circuit multiple times and averaging over the measurement results.\footnote{Note that the quantum model is a theoretical construction, since physical measurements will always result in an \textit{estimate} of the output expectation, making $f$ a random variable -- a complication that we will ignore here.} The quantum circuit itself is constructed from $L$ \textit{layers}, each consisting of a data encoding circuit block $S(x)$ and a trainable circuit block $W(\btheta)$ controlled by the parameters $\btheta$ (see Fig. \ref{fig:scheme}). The data encoding block is the same in every layer and consists of gates of the form $\mathcal{G}(x) = e^{-ix H}$, where $H$ is a Hamiltonian that \textit{generates} the ``time evolution'' used to encode the data. Since we want to focus on the role of the data encoding, and to avoid further assumptions on how the trainable circuit blocks are parametrised, we view the trainable circuit blocks as arbitrary unitary operations, $W(\btheta)=W$, and drop the subscript of $f_{\btheta}$ from here on.\footnote{Of course, in realistic near-term settings these unitaries are implemented as short gate sequences and are by no means universal, and there are many interesting questions around how a specific parametrisation influences the properties of the resulting quantum model.}
With this assumption, the overall quantum circuit has the form 
\begin{equation}
	U(x) = W^{(L+1)} S(x) W^{(L)} \dots   W^{(2)} S(x)W^{(1)}.
	\label{eq:unitary}
\end{equation}
Note that the encoding strategy is very natural, since the physical control parameters of quantum dynamics usually enter as time evolutions of Hamiltonians -- the most prominent example being Pauli rotations. This model includes ``parallel encodings'' that repeat the encoding on different subsystems~\cite{rebentrost2014quantum}, as well as ``data reuploading'', where the encoding is repeated
multiple times in sequence~\cite{perez2020data} (see Fig.~\ref{fig:parallel_vs_sequence}). With a small amount of classical pre-processing this model includes even many quantum machine learning algorithms that are not based on the principles of parametrised circuits (see also Section \ref{Sec:hybrid}).
 
 \begin{figure}[t]
    \centering
    \includegraphics[width=0.45\textwidth]{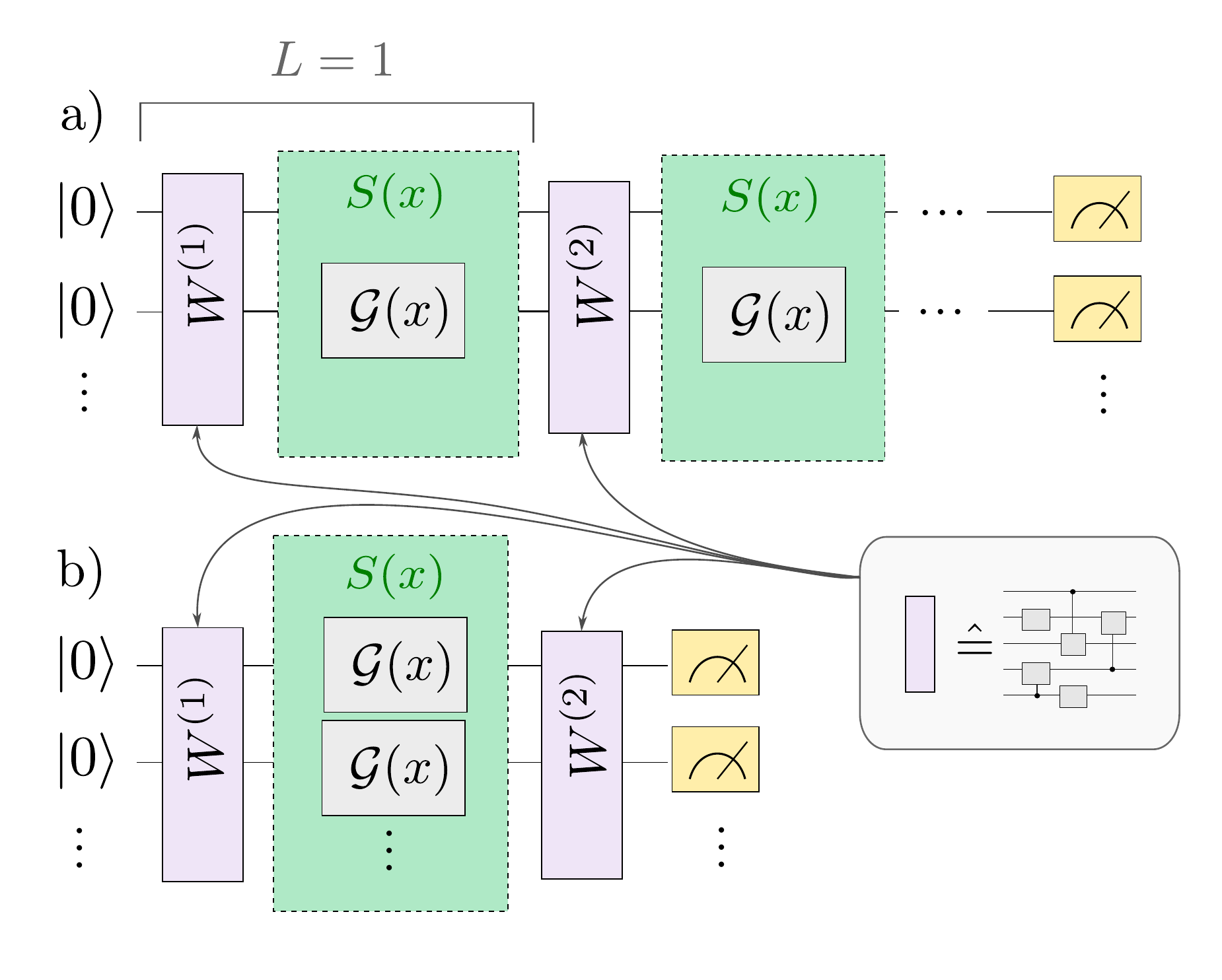}
    \caption{The general quantum model considered in this paper includes qubit-based circuits where the encoding subroutine consists of a single-qubit gate $\mathcal{G}(x)$, which is often used in practice. The picture illustrates two special cases investigated in Section \ref{Sec:expressivity}: (a) shows a circuit where the scalar input feature $x$ is encoded by one single-qubit gate, which can be repeated $r = L>1$ times but always acts on the same qubit, and (b) repeats the encoding gate $r$ times in ``parallel'' using only one layer. Note that the trainable blocks $W$ (purple rectangles) represent arbitrary unitaries, which in practice would be implemented as a sequence of local gates (inset).}
    \label{fig:parallel_vs_sequence}
\end{figure}

Our goal is to write $f$ as a partial Fourier series 
\begin{equation}
    f(x) = \sum_{n \in \Omega} c_{n} e^{i n x},
    \label{eq:partial_fourier_series}
\end{equation}
with integer-valued frequencies (if $\Omega = \{-K,\ldots, K\}$, then we call
\eqref{eq:partial_fourier_series} a \textit{truncated Fourier series}). The first step is to note that one can always find an eigenvalue decomposition of the generator Hamiltonian $H = V^{\dagger} \Sigma V$ where $\Sigma$ is a diagonal operator containing $H$'s eigenvalues $\lambda_1, ..., \lambda_d$  on its diagonal. The data encoding unitary becomes $S(x) = V^{\dagger} e^{-i x \Sigma} V$, and we can ``absorb'' $V$, $V^{\dagger}$ into the arbitrary unitaries $W' = VWV^\dagger$. Hence, without loss of generality we will assume that $H$ is diagonal. This allows us to separate the data-dependent expressions from the remainder of the circuit in each component $i$ of the quantum state $U(x) \ket{0}$,
\begin{multline}
    \left[U(x) \ket{0}\right]_{i} =  \sum_{j_1 \dots j_L =1}^d e^{-i(\lambda_{j_1} + \dots + \lambda_{j_L}) x} \\
    \times W_{i j_L}^{(L+1)} \dots W_{j_2 j_1}^{(2)} W_{j_1 1}^{(1)} .
    \label{eq:circuit_fourier}
\end{multline}
For ease of notation we introduce the multi-index $\boldsymbol{j} = \{j_1, \dots, j_L\} \in [d]^L$, where $[d]^L$ denotes the set of any $L$ integers between $1, \dots, d$. We can then denote the sum of eigenvalues for a given $\boldsymbol{j}$ by $\Lambda_{\boldsymbol{j}} = \lambda_{j_1} + \dots + \lambda_{j_L}$, and write 
\begin{multline}
    \left[U(x) \ket{0}\right]_{i} =  \sum_{\boldsymbol{j} \in [d]^L} e^{-i \Lambda_{\boldsymbol{j}} x} W_{i j_L}^{(L+1)} \dots W_{j_2 j_1}^{(2)} W_{j_1 1}^{(1)} .
    \label{eq:circuit_fourier2}
\end{multline}
To consider the full quantum model from Eq.~\eqref{eq:model} we need to take into account the complex conjugation of this expression as well as the measurement, and get
\begin{equation}
    f(x) = \sum_{ \boldsymbol{k}, \boldsymbol{j} \in [d]^L} e^{i (\Lambda_{\boldsymbol{k}} - \Lambda_{\boldsymbol{j}})  x} a_{\boldsymbol{k}, \boldsymbol{j}},
    \label{eq:raw_sum}
\end{equation} 
where the $a_{\boldsymbol{k}, \boldsymbol{j}}$ contain the terms stemming from the arbitrary unitaries and measurement,
\begin{multline} a_{\boldsymbol{k}, \boldsymbol{j}} = \sum_{i, i'}  (W^*)_{1 k_1}^{(1)} (W^*)_{j_1 j_2}^{(2)}  \dots (W^*)_{ j_L i}^{(L+1)} M_{i, i'} \\
\times W_{i' j_L}^{(L+1)} \dots W_{j_2 j_1}^{(2)} W_{j_1 1}^{(1)}.
\end{multline}
The second step consists of grouping all terms in the sum \eqref{eq:raw_sum} whose basis function $e^{i (\Lambda_{\boldsymbol{k}} - \Lambda_{\boldsymbol{j}})  x}$ have the same frequency $\omega = \Lambda_{\boldsymbol{k}} - \Lambda_{\boldsymbol{j}}$. All frequencies accessible to the quantum model are contained in its frequency spectrum
\begin{equation} 
\Omega = \{\Lambda_{\boldsymbol{k}} - \Lambda_{\boldsymbol{j}},  \; \boldsymbol{k}, \boldsymbol{j} \in [d]^L \}.
\label{eq:fourier_omega}
\end{equation}
This yields
\begin{equation}
    f(x) = \sum_{\omega \in \Omega} c_{\omega} e^{i \omega x}
    \label{eq:fourier_sum}
\end{equation} 
where the coefficients are obtained by summing over all $a_{\kk,\jj}$ contributing to the same frequency
\begin{align}
c_{\omega} = \sum_{\substack{\boldsymbol{k},\boldsymbol{j} \in [d]^L\\\Lambda_{\boldsymbol{k}}-\Lambda_{\boldsymbol{j}} = \omega}} a_{\boldsymbol{k}, \boldsymbol{j}}.
    \label{eq:model_fourier_coefficients}
\end{align}

We note that the frequency spectrum $\Omega$ has the following important properties: $0 \in \Omega$, and for every frequency $\omega\in\Omega$, we have also that $-\omega \in \Omega$. Additionally, since
$c_\omega = c^*_{-\omega}$, Eq.~\eqref{eq:fourier_sum} realises a real-valued function. We will therefore denote with $K = (|\Omega|-1)/2$ the \emph{size} of the spectrum, as it quantifies how many independent non-zero frequencies the model has access to. The largest available frequency $D = \max(\Omega)$ is called the \emph{degree} of the spectrum. Furthermore, the coefficients $c_{\omega}$ are determined by the arbitrary gates $W^{(1)} \dots W^{(L+1)}$ (which absorbed the $V$, $V^{\dagger}$ from the encoding Hamiltonians), as well as by the measurement observable. As a consequence, a quantum model's frequency spectrum is solely determined by the eigenvalues of the data encoding gates, while its Fourier coefficients depend on the entire circuit, as was claimed in Fig.~\ref{fig:scheme}. 
While so far we have not imposed restrictions on the frequencies $\omega$, one can see that for integer-valued eigenvalues $\lambda_1, \dots, \lambda_d$, the frequencies in $\Omega$ are themselves integer-valued, and Eq.~\eqref{eq:fourier_sum} yields the real-valued partial Fourier series from Eq.~\eqref{eq:partial_fourier_series}. As we will show in the following section, common data encoding strategies in near-term quantum machine learning fulfill the property of an integer-valued frequency spectrum. Even if the eigenvalues of the encoding gate generators, and therefore the accessible frequencies $\omega \in \Omega$, are merely integer-valued multiples of a ``base frequency'' $\{n_1 \omega_0, n_2 \omega_0, \dots\}$, the treatment is still analogous to the integer case (see Appendix~\ref{app:non_integer_frequencies}). We therefore focus much of our analysis on this case.

For both integer or non-integer frequencies, the expressivity of a quantum model is determined by two different properties: the frequency spectrum of the quantum model, including its size and degree, and the expressivity of the coefficients controlled by the model. As we will show in the next section, these two properties give us insights into the function classes that different quantum models can learn. 

\section{The expressivity of quantum models}\label{Sec:expressivity}

We proceed to use the Fourier series formalism to investigate the expressivity of quantum models. We start with an analysis of the popular strategy \cite{mitarai2018quantum, schuld2020circuit, havlivcek2019supervised, hubregtsen2020evaluation, koide2020quantum, watabequantum, blank2019quantum, zhao2019building} of using single-qubit Pauli rotations in the encoding subroutine $S(x)$ in order to showcase the practical value of the approach. We then characterise the limits of a quantum model's expressivity for a given data encoding gate in more general terms. 

\subsection{A single Pauli-rotation encoding can only learn a sine function}\label{Sec:sine}

As a ``warm-up" application of the Fourier series formalism, we start by considering a simple quantum model with $L=1$, where we use a single-qubit gate $\mathcal{G}(x) = e^{-i x H}$ to encode the input $x$ into the circuit (see also Fig.~\ref{fig:parallel_vs_sequence}a with $L=1$),
\begin{equation}
    U(x) = W^{(2)}\mathcal{G}(x) W^{(1)}.
    \label{eq:single_qubit_encoding}
\end{equation}
As a single-qubit gate generator, $H$ has two distinct eigenvalues  $(\lambda_1,\lambda_2)$. We can without loss of generality always rescale the energy spectrum to $(- \gamma, \gamma)$ because the global phase is unobservable. We note that the class of such encoding gates includes Pauli rotations, with $H = (1/2)\sigma$ for $\sigma\in \{\sigma_{x},\sigma_y,\sigma_z\}$, for which $\gamma=\frac{1}{2}$. We aim to show that models of the type \eqref{eq:single_qubit_encoding} \textit{always} lead to functions of the form $f(x) = A \sin(2\gamma  x + B) + C$ where $A, B, C$ are constants determined by the non-encoding part of the variational circuit, which reproduces the prior observation from \cite{ostaszewski2019quantum}. A sine function can be described by a truncated Fourier series of degree $1$ -- and in the next section we will go on to show how one can systematically increase the degree by repeating the encoding gate. 

First, since we can absorb the factor $\gamma$ into the data input by re-scaling it via $\tilde{x} = \gamma x$, we can assume without loss of generality that the eigenvalues of $H$ are always $\lambda_1 = -1, \lambda_2 = 1$. From Eq.~\eqref{eq:fourier_omega} we can immediately see that the spectrum of the quantum model is given by $\Omega = \{-2, 0, 2\}$ (since the possible differences $\lambda_{k_1} - \lambda_{j_1}$ for $\lambda_{k_1}, \lambda_{j_1} \in \{-1, 1\}$ are $-1-(1)$, $-1-(-1)$, $1-(1),$ and $1-(-1)$).
The Fourier coefficients in Eq.~\eqref{eq:model_fourier_coefficients} become
\begin{align}
    c_{0} &=  \sum_{i, i'} M_{i i'} (W^*)^{(1)}_{1 2} (W^*)^{(2)}_{2 i} W^{(2)}_{i' 1} W^{(1)}_{1 1},\\
    c_{2} &=  \sum_{i,i'} M_{i i'} (W^*)^{(1)}_{1 1} (W^*)^{(2)}_{1 i} W^{(2)}_{i' 2} W^{(1)}_{2 1},\\
    c_{-2} &=  c^*_{2},
    \label{eq:model_fourier_single_Pauli}
\end{align}
and the quantum model's frequency spectrum consists of a single non-zero frequency:
\begin{align*}
    f(x) &=  c_{-2}  e^{i 2  \tilde{x}} + c_{0}  + c_{2}  e^{-i 2  \tilde{x}} \\
    &= c_{0}  +  2|c_{2}| \cos (2\tilde{x} - \operatorname{arg}(c_2)),
\end{align*}
where $\operatorname{arg}(c_2)$ is the complex phase of $c_2$. For Pauli rotations, one has $\tilde{x} = \gamma x = \frac{x}{2}$, and we recover the result of \cite{ostaszewski2019quantum} with $A = 2 |c_{2}|, B = -\pi/2 - \operatorname{arg}(c_2),$ and $C = c_{0}$.
Importantly, we have not assumed anything about the number of qubits, the nature of the unitaries $W$, or the measurement $M$. This illustrates a key point of this paper: even with the ability to implement very wide and deep quantum circuits (which may even be classically intractable to simulate), \textit{the expressivity of the corresponding quantum model is fundamentally limited by the data encoding strategy}.

To support this finding, Fig.~\ref{fig:fitting} shows numerical evidence: encoding data via a Pauli-X rotation results in a quantum model that can only learn to fit a Fourier series of a single frequency -- and only if that frequency is exactly matched by how the data is scaled.  

\begin{figure}[t]
    \centering
    \includegraphics[width=0.4\textwidth]{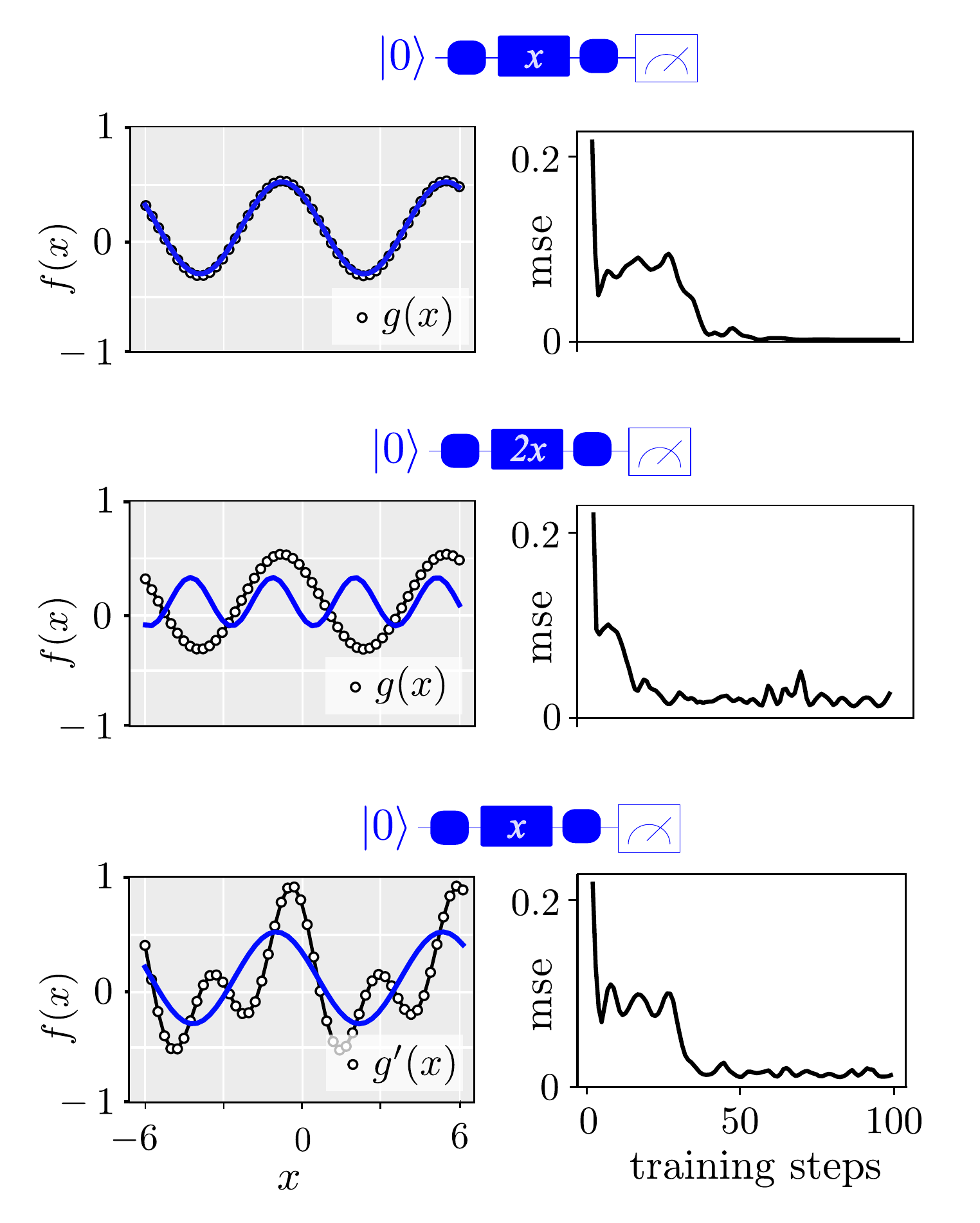}
    \caption{A parametrised quantum model is trained with data samples (white circles) to fit a target function $g(x) = \sum_{n=-1}^1 c_{n} e^{-nix}$ or $g'(x) = \sum_{n=-2}^2 c_{n} e^{-nix}$ with coefficients $c_0=0.1$, $c_1 = c_2 = 0.15 - 0.15i$. The variational circuit is of the form $f(x) = \bra{0} U^{\dagger}(x) \sigma_z U(x) \ket{0} $ where $\ket{0}$ is a single qubit, and $U = W^{(2)} R_x(x) W^{(1)} $. The $W$ (round blue symbols) are implemented as general rotation gates parametrised by three learnable weights each, and $R_x$ (square blue symbols) is a single Pauli-X rotation. The left panels show the quantum model function $f(x)$ and target function $g(x), g'(x)$, while the right panels show the mean squared error between the data sampled from $g$ and $f$ during a typical training run. Feeding in the input $x$ as is (top row), the quantum model easily fits the target of degree $1$. Re-scaling the inputs $x \rightarrow 2x$ causes a frequency mismatch, and the model cannot learn the target any more (middle row). However, even with the correct scaling, the variational circuit cannot fit the target function of degree $2$ (bottom row). The experiments in this paper were all performed using the \textit{PennyLane} software library \cite{bergholm2018pennylane}.}
    \label{fig:fitting}
\end{figure}

\subsection{Repeated Pauli encodings linearly extend the frequency spectrum}\label{Sec:repeat}

Given the severe limitations exposed in the previous section, a natural question is how we can extend the accessible frequency spectrum of a quantum model. To this end, we demonstrate in this section that by using either single-layer models with $L=1$ where the encoding gate is repeated $r$ times in \textit{parallel} (as per Fig.~\ref{fig:parallel_vs_sequence}b), or multi-layer models with $L>1$ where the encoding gate is effectively repeated $r=L$ times in \textit{series} (as per Fig.~\ref{fig:parallel_vs_sequence}a), one can systematically increase the degree of the truncated Fourier series to $r$. We note once again that both of these techniques have been utilised in prior practical applications \cite{schuld2020circuit, mitarai2018quantum, hubregtsen2020evaluation, koide2020quantum, watabequantum, perez2020data}, and as such the observations we make here offer insight into the properties of these models. 

Firstly, let us consider the case of single-qubit Pauli rotations repeated in parallel (Fig.~\ref{fig:parallel_vs_sequence}b). This is a special case of our base model in Eq.~\eqref{eq:model}, with $L=1$, and
\begin{align}
    S(x) &= e^{-i\frac{x}{2}\sigma_r}\otimes\ldots\otimes e^{-i\frac{x}{2}\sigma_1},\\
    &:= e^{-ixH}
\end{align}
where $\sigma_j\in \{\sigma_x, \sigma_y, \sigma_z\}$. The fact that all rotation gates commute (as they act on different qubits) allows us to diagonalise $H$ by diagonalising each rotation gate individually. Doing this, we find that
\begin{align}
    S(x) &= V_re^{-i\frac{x}{2}\sigma_z}V_r^{\dagger}\otimes\ldots\otimes V_1e^{-i\frac{x}{2}\sigma_z}V_1^{\dagger},\\
    &= V\mathrm{exp}\left(-i\frac{x}{2}\sum_{q = 1}^r\sigma^{(q)}_z\right)V^{\dagger},\\
    &:= Ve^{-ix\Sigma}V^{\dagger},
\end{align}
where $\sigma^{(q)}_z$ is the (diagonal) $r$-qubit operator which acts non-trivially, via $\sigma_z$, only on the $q$'th qubit. Performing the calculation yields $\Sigma = \rm{diag}\left( \lambda_1, \dots, \lambda_{2^r}\right)$, with the $r+1$ unique entries
\begin{align*}
\lambda_p = \left(\frac{p}{2} - \frac{r-p}{2}\right) = p - \frac{r}{2}, \; p \in \{0,\dots, r\}, 
\end{align*}
which are all possible sums of $r$ values $\pm 1/2$. 
According to Eq.~\eqref{eq:fourier_omega}, the frequency spectrum for $L=1$ contains differences of any two of these eigenvalues, and we get
\begin{align}
    \Omega_{\rm par} & = \left\{ \lambda_{k_1} - \lambda_{j_1} | \; k_1, j_1 \in \{1,\dots,2^r\right\} \} \\
    &= \left\{\left(p - \frac{r}{2}\right) - \left(p' - \frac{r}{2}\right) \right| \nonumber\\
     & \phantom{AAAAAAAAAAA} \; p, p'\in \{0,\dots,r \} \}\\ 
    &= \left\{ p-p'\, | \, p, p'\in \{0,\dots,r \} \right\}\\
    & =\{-r, -(r-1), \dots, 0 , \dots, r-1, r\}.
\end{align}
Hence, a univariate quantum model with $r$ parallel Pauli-rotation encodings can be expressed as a truncated Fourier series of degree $r$. 

Interestingly, the same scaling effect is achieved by a single-qubit Pauli rotation encoding repeated layer-wise (Fig.~\ref{fig:parallel_vs_sequence}a). 
Consider the quantum model in Eqs.~\eqref{eq:model} and \eqref{eq:unitary}, for  $L=r> 1$ layers, where $S(x) = \mathrm{exp}(-i(x/2)\sigma_j)$ is a single-qubit Pauli rotation (i.e. $\sigma_j \in \{\sigma_x,\sigma_y,\sigma_z\}$) which acts on the same qubit in each layer.
The circuit in Eq.~\eqref{eq:unitary} becomes 
$$U(x) = W^{(L+1)}  e^{-i\frac{x}{2}\sigma_L} W^{(L)}\dots W^{(2)} e^{-i\frac{x}{2}\sigma_1} W^{(1)}.$$ Diagonalizing the Pauli rotations as before, then gives us $\Sigma = (1/2)\sigma_z$ for all encoding layers.
The frequency spectrum from Eq.~\eqref{eq:fourier_omega} is a sum of $2r$ terms of value $\pm 1/2$,
\begin{multline}
    \Omega_{\rm seq} = \{(\lambda_{k_1} + \dots + \lambda_{k_r}) - (\lambda_{j_1} + \dots + \lambda_{j_r}) \, |\\
    k_1,\dots, k_r, j_1,\dots, j_r \in \{1, 2 \} \}.
\end{multline} 
After a short calculation, one finds that $\Omega_{\rm seq} = \Omega_{\rm par}$. Again, a quantum model with $r$ sequential repetitions of the single-qubit Pauli encoding can be expressed as a truncated Fourier series of degree $r$. The growth mechanism of a quantum model's frequency spectrum via parallel and sequential repetitions of single-qubit Pauli encodings is numerically illustrated in Fig.~\ref{fig:L1-5}.

\begin{figure}
    \centering
    \includegraphics[width=0.47\textwidth]{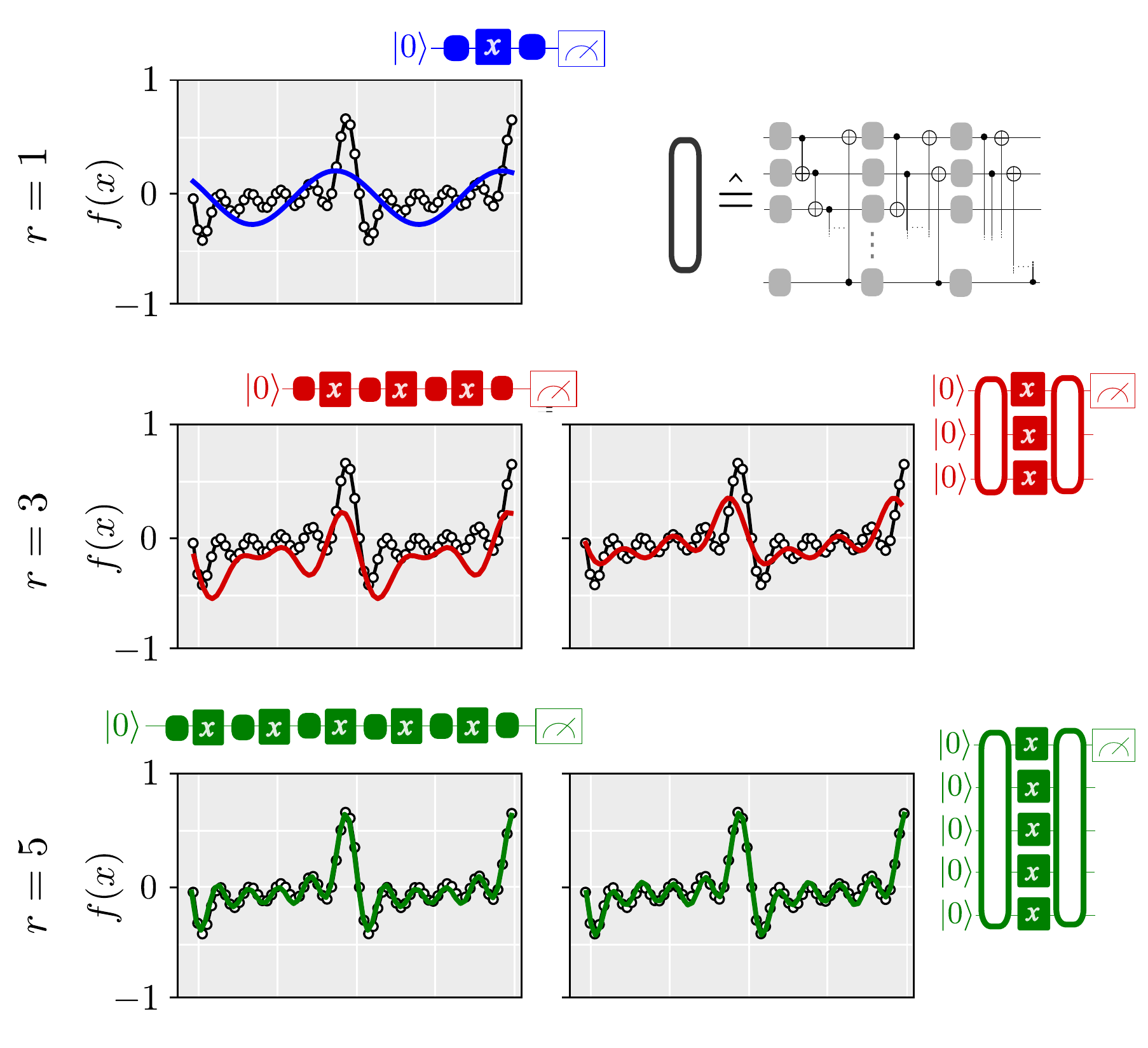}
    \caption{Fitting a truncated Fourier series of degree $5$, $g(x) = \sum_{n=-5}^5 c_n e^{2 i n x}$ with $c_n = 0.05 - 0.05i$ for $n=1,\dots,5$ and $c_0=0$, using a quantum model that repeats the encoding $r =1, 3, 5$ times in sequence (left) and in parallel (right). Increasing $r$ allows for closer and closer fits until $r=5$ fits the data almost perfectly in both cases - illustrating that parallel and sequential repetitions of Pauli encodings extend the Fourier spectrum in the same manner. All models were trained with at most 200 steps of an Adam optimiser with learning rate $0.3$ and batch size $25$. For the ``parallel'' simulations, the $W$ are not arbitrary unitaries but implemented by a smaller ansatz of three layers of parametrised rotations as well as entangling CNOT gates, as per Ref.~\cite{schuld2020circuit}, which is depicted by the hollow rounded gate symbols. The quantum model still easily fitted the target function, which suggests that the results of this paper are of relevance for realistic quantum models.}
    \label{fig:L1-5}
\end{figure}

\subsection{Limits of expressivity}\label{Sec:limits}

The representation of quantum models as Fourier-type sums
immediately allows us to derive upper bounds on the expressivity of such quantum models when using $L$ repetitions of an encoding gate of dimension $d$ (which is at most the size of the overall Hilbert space). Firstly, let us consider the maximum spectrum size $K(L, d)$ of a quantum model, quantifying the number of frequencies it can ``support'' or ``has access to''. Since the frequency spectrum is defined as $\Omega = \{(\lambda_{j_1} + \dots \lambda_{k_L}) - (\lambda_{j_1} + \dots + \lambda_{k_L})\}$ (where the indices $j_1,\dots, j_L$, $k_1, \dots, k_L$ run over all dimensions of the encoding gate, from $1$ to $d$), the frequencies are sums of $2L$ terms, each having $d$ potential values. As a result, they can at most realise $d^{2L}$ distinct values -- irrespective of whether the eigenvalues are real or integer-valued. Since the size $K$ counts the pairs $-\omega, \omega \in \Omega$ as one and excludes the ``zero frequency", we get
\begin{equation}
    K \leq \frac{d^{2L}}{2}-1.
    \label{eq:bound}
\end{equation}
As an example, if data is encoded in a single-qubit encoding gate, we recover the result from the previous sections where the model has degree $\frac{2^2}{2}-1 = 1$. Using $L$ different encoding gates increases this to $\frac{2^{2L}}{2}-1$. As we have seen, further assumptions on the eigenvalues allow us to make this bound a lot tighter; for example when the $L$ repetitions use \textit{the same} single-qubit encoding gate, $K=L$.

An interesting question is whether there is a single quantum gate which can encode data into a quantum model that supports the frequency spectrum $\Omega_{\infty} = \{- \infty, \dots, -1, 0, 1, \dots, \infty\}$ of a full Fourier series. The answer is yes: the ubiquitous phase shifts in continuous-variable (CV) quantum systems, which correspond to a free evolution of a harmonic oscillator, have the number operator $\hat{n} = \operatorname{diag}(0, 1, 2, \dots)$ as a generator. 

\begin{figure*}
    \centering
    \includegraphics[width=0.9\textwidth]{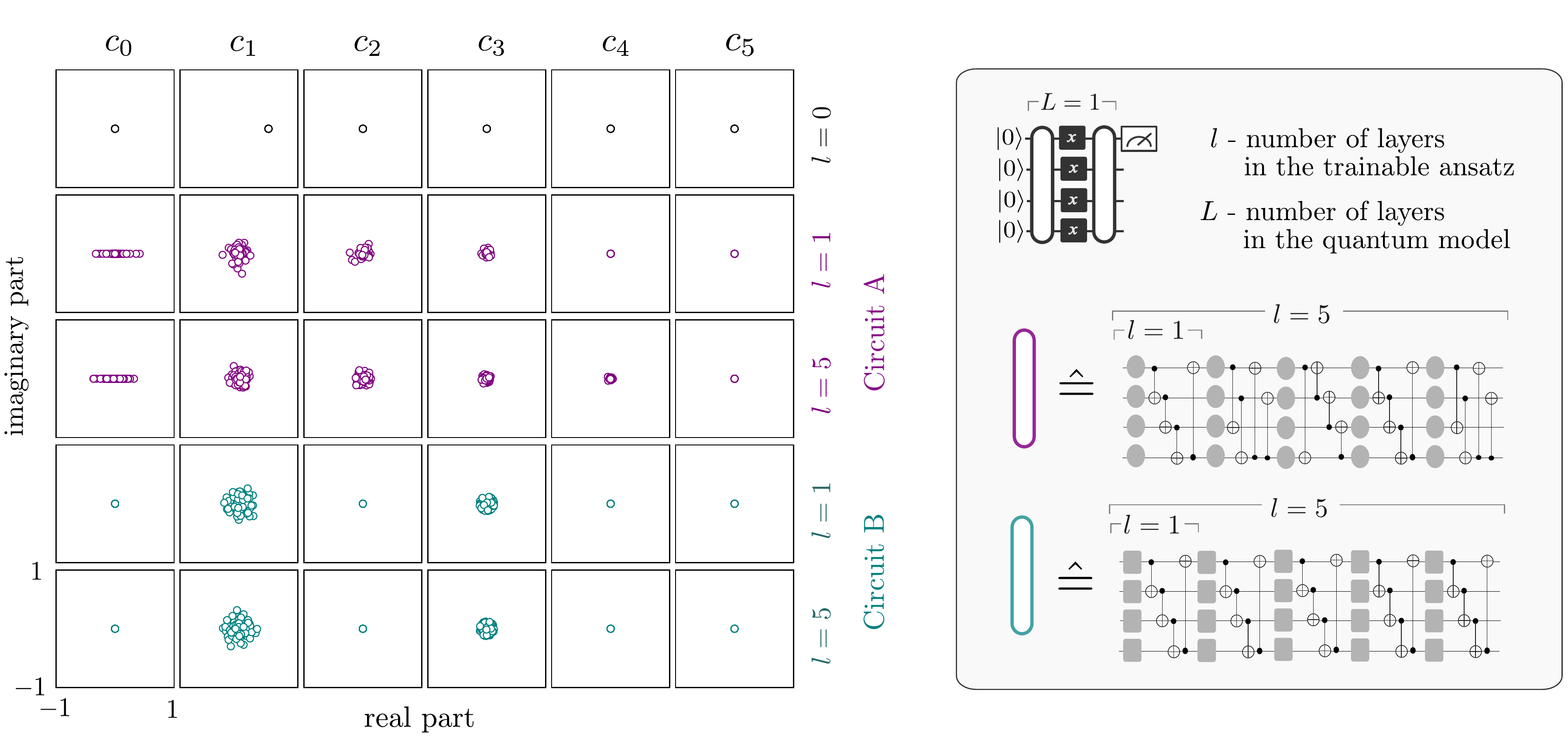}
    \caption{Real and imaginary parts of the first six Fourier coefficients sampled from $100$ randomly initialised $L=1$ quantum models. The models share the same encoding strategy of parallel Pauli-X rotations (square symbols) but vary in the ansatz and number of layers for the trainable unitaries $W$. Circuit A uses an ansatz of trainable arbitrary single qubit rotations and layer-dependent entangling structure proposed in \cite{schuld2020circuit} and already used in Fig.~\ref{fig:L1-5}, while Circuit B uses trainable Pauli-X rotations with a simple entangling structure. The plots suggest that the ``expressivity'' of the trainable circuit block --  here represented by increasing the number of times $l$ an ansatz is repeated -- has little influence on the distribution of the Fourier coefficients, as opposed to the type of ansatz.}
    \label{fig:coeffs}
\end{figure*}

While the frequency spectrum of a quantum model can directly be derived from the input encoding gates, the flexibility in the coefficients is a lot harder to investigate systematically (we will do so for special cases in the universality proofs in Section \ref{Sec:universality}). In principle, every block $W^{(1)}, \dots, W^{(L+1)}$, as well as the measurement observable, contribute to every Fourier coefficient. This means that only a few degrees of freedom in the gates may change an exponentially large (or, in the case of continuous-variable quantum computing, infinite) amount of Fourier coefficients. However, these Fourier coefficients are not arbitrary, but functions of the limited degrees of freedom of the quantum circuit, and a quantum circuit of a certain structure may only be able to realise a small subset of the entire set of all possible Fourier coefficients $\{c_{n}\}$. To arbitrarily control $K+1$ complex Fourier coefficients, we need \textit{at least} $M \geq 2K+1$ real degrees of freedom -- in other words, parameters $\btheta = (\theta_1, \dots, \theta_M)$ -- in the quantum circuit. As a special case, we saw that repeating a Pauli encoding $L$ times supports a spectrum of size $L$, which means that we need at least $2L$ degrees of freedom in the quantum circuit to control the Fourier coefficients arbitrarily -- a scaling that is realistic for shallow circuits to ``utilise the full power'' of the frequency spectrum.

While a systematic analysis of how a parametrised ansatz for the trainable blocks $W$ impacts the control of a quantum model's Fourier coefficients exceeds the scope of this paper, our simulations suggest that even quantum models with shallow trainable circuit blocks $W$ give rise to rich subsets of Fourier coefficients (see Fig.~\ref{fig:coeffs}). However, as the figure shows, an ansatz may structurally set a certain Fourier coefficient to zero. An interesting further observation is that the variance of the coefficients decreases with higher orders. Mathematically, this property stems from the fact that the number of terms in the sum of Eq. \eqref{eq:model_fourier_coefficients} tends to decrease with larger frequencies, since there are fewer ways to construct those frequencies by the difference $\Lambda_{\jj} - \Lambda_{\kk}$ of sums of encoding generator eigenvalues. We note that the Fourier coefficients of square-integrable functions show a similar behaviour, which contributes to the convergence of such series.

\section{Quantum models are asymptotically universal}\label{Sec:universality}

In the previous sections we have seen, at least for univariate functions, that certain quantum models can be written as partial Fourier series, in which the accessible frequencies are fully determined by the spectra of the Hamiltonians generating the data-encoding gates. Additionally, by using Pauli rotations as an explicit example, we have shown that by repeating such encodings, either in parallel $(L=1)$ or in series ($L>1$), it is possible to realise a \textit{truncated} Fourier series, with the number of accessible frequencies determined by the number of data-encoding gate repetitions. In light of these results, it is clear that if we allow for sufficiently many repetitions of simple data-encoding gates (such as Pauli rotations), or for Hamiltonians with large enough dimension and suitably non-degenerate spectra, then quantum models can realise arbitrary frequency spectra.

However, as discussed in the previous section, the expressivity of a quantum model is determined not only by the accessible frequency spectrum, but also by the flexibility one has in adjusting the contributions of the frequencies, i.e., with which flexibility the Fourier coefficients can be chosen. In this section we show that if one allows for trainable circuit blocks which are flexible enough to realise arbitrary global unitaries, then there exists an $L=1$ quantum model which can realise all possible sets of Fourier coefficients. Combined with the observations from the previous sections, this allows us to show that such quantum models are asymptotically universal, in the sense that if we allow the global Hilbert space dimension (or the number of finite dimensional subsystems) to tend to infinity, then such a quantum model can approximate, to arbitrary accuracy, any square-integrable function on a suitable domain. 

More specifically, we consider the (multivariate) single layer quantum model $f_{\vec{\theta}}\colon\mathbb{R}^N\rightarrow \mathbb{R}$ defined via
\begin{align}
    f_{\vec{\theta}}(\vec{x}) = \langle 0 | U^\dagger(\vec{\theta},\vec{x}) M U(\vec{\theta},\vec{x}) |0\rangle,
\end{align}
where
\begin{equation}
    U(\vec{\theta},\vec{x}) = W^{(2)}(\vec{\theta}^{(2)})S(\vec{x}) W^{(1)}(\vec{\theta}^{(1)}),
\end{equation}
with $\vec{\theta}^{(1)}, \vec{\theta}^{(2)} \subseteq \vec{\theta}$ and
\begin{equation}
    S(\vec{x}) := e^{-i x_1 H_1}\otimes\ldots\otimes e^{-i x_N H_N }.
\end{equation}
The above model is a natural extension of the univariate $L=1$ model we explored in previous sections. In Appendix  \ref{app:multivariate_FS} we show that it naturally realises a multivariate Fourier series, with the frequency spectrum fully determined by the spectra of the data-encoding Hamiltonians $\{H_l\}$, and the Fourier coefficients determined by the remainder of the circuit.

It is important to emphasise that in practical applications one would typically consider trainable circuit blocks whose circuit depth scales in a controlled way with respect to the number of qubits in the circuit. However, we will in this work assume that the trainable circuit blocks are sufficiently flexible to realise arbitrary global unitaries, which may require exponential circuit depth when decomposed into natural primitive gate sets. Given this, the asymptotic universality of quantum models with either constant, logarithmic or polynomial circuit depth trainable blocks remains an interesting open question. 

With this assumption on the trainable circuit blocks, we can drop the explicit dependence on $\vec{\theta}$, and by absorbing $W^{(1)}$ into the initial state $\ket{\Gamma}$, and $W^{(2)}$ into the observable $M$, consider instead the equivalent model
\begin{align}\label{eq:multivariate_model}
    f(\vec{x}) = \langle \Gamma | S^{\dagger}(\vec{x}) MS(\vec{x}) |\Gamma\rangle,
\end{align}
where the universality of $W^{(1)}$ and $W^{(2)}$ is reinterpreted as the assumption that $|\Gamma\rangle$ can be an arbitrary state, and $M$ an arbitrary observable. In order to simplify things further, we will also make the additional assumption that all data-encoding Hamiltonians are equal -- i.e., 
that 
\begin{align}
    S(\vec{x}) &:= e^{-ix_1 H}\otimes\ldots\otimes e^{-i x_N H}\\
    &:=S_H(\vec{x}).
\end{align}
We are interested in a reasonable notion of universality in the \emph{asymptotic} regime of infinitely many available subsystems. To formalise this, we introduce the concept of a \emph{Hamiltonian family} $\{H_m \, | \, m \in \bbN \}$ where $H_m$ acts on $m$ subsystems of dimension $d$. An explicit example of such a family is a simple tensor product of Pauli rotations, as studied in Section~\ref{Sec:repeat}, which corresponds to the Hamiltonian
\begin{align}\label{eq:pauli_onsite}
    H_m = \sum_{i=1}^m \sigma_q^{(i)}.
\end{align}
As illustrated in Fig.~\ref{fig:universality}, such a Hamiltonian family defines a family of models $\{f_{m}\}$ via 
\begin{equation}\label{eq:indexed_model}
    f_m(\vec{x}) = \langle \Gamma | S_{H_m}^{\dagger}(\vec{x}) MS_{H_m}(\vec{x}) |\Gamma\rangle,
\end{equation}
where for each $m$, the measurement $M$ and the state $|\Gamma\rangle$ (or equivalently the unitaries $W^{(1)}$ and $W^{(2)}$) are the learnable elements of the model.

\begin{figure}
    \centering
    \includegraphics[width=0.45\textwidth]{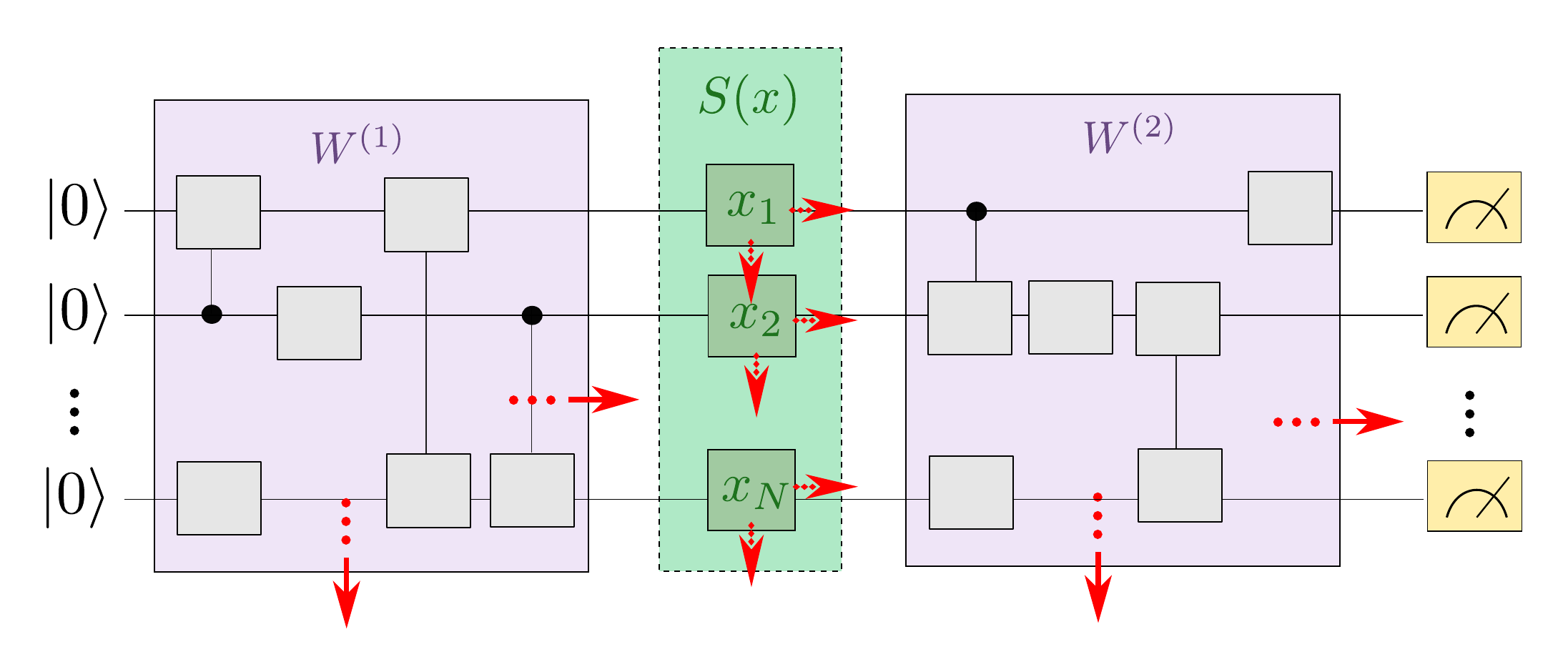}
    \caption{Multivariate $L=1$ quantum model considered for the universality theorem. Here, $S(\bx)$ consists of feature-encoding gates acting on different subsystems (green boxes). The Hamiltonians that generate these gates are defined to increase the ``richness'' of their spectrum with growing dimension of the subsystems (red arrows). Since we assume that the circuit depth and structure of the trainable unitaries $W^{(1)}$ and $W^{(2)}$ is sufficient to allow for the realisation of arbitrary unitary operations, the trainable circuits grow in dimension along with the total system size.}
    \label{fig:universality}
\end{figure}

Now, given some Hamiltonian $H_m$ with eigenvalues $\{\lambda_1,\ldots,\lambda_{d^m}\}$, we call 
\begin{equation}
    \Omega_{H_m} = \{\lambda_j-\lambda_k\,|\, j,k \in \{1,\ldots,d^m\}\}\label{eq:freq_spec_ham}
\end{equation}
the \textit{frequency spectrum} associated with $H_m$. To achieve universality, we need a Hamiltonian family whose frequency spectrum asymptotically contains any integer frequency.
We formalise this via the following notion: a Hamiltonian family $\{ H_m \}$ is a \emph{universal Hamiltonian family} if it has the property that for all $K\in \mathbb{N}$ there exists some $m\in \mathbb{N}$ such that
\begin{equation}
    \bbZ_K =  \{-K,\ldots,0,\ldots,K\} \subseteq \Omega_{H_m}.
\end{equation}

As we have seen in the previous section, the Hamiltonian family defined by the Hamiltonians in Eq.~\eqref{eq:pauli_onsite} is indeed a  universal Hamiltonian family, with $m=K$. 
As the possible number of frequencies grows exponentially, one could think of more complicated Hamiltonian families in which the required number of available subsystems only grows logarithmically $m = O(\log K)$, at the cost of more complicated, global Hamiltonian terms.
With this setup, we can now state the following universality result:

\begin{theorem*}\label{thm:universality}
Let $\{H_m\}$ be a universal Hamiltonian family, and $\{f_m\}$ the associated quantum model family, defined via Eq.~\eqref{eq:indexed_model}.  For all functions $g \in L_2([0,2\pi]^N)$, and for all $\epsilon > 0$, there exists some $m'\in \mathbb{N}$, some state $|\Gamma\rangle \in \mathbb{C}^{d^{m'}}$, and some observable $M$ such that
\begin{equation}
    \lvert\rvert f_{m'} - g \lvert\rvert_2 \leq \epsilon.
\end{equation}
\end{theorem*}
A full proof is given in Appendix \ref{app:thm_proof}, however in the following we will provide a sketch of the proof in order to give an outline of the ideas and techniques. 
The proof begins by noting that any square-integrable function $g$ on a finite interval can be approximated by a truncated Fourier series to arbitrary precision. We therefore reduce the task to finding a quantum model for this truncated Fourier series. The universality property of the Hamiltonian family implies that the multivariate models we consider can express all necessary frequencies to perform that approximation. We then show how to use the freedom in choosing the initial state and the observable to reproduce the truncated Fourier series of $g$ exactly, leading to an approximation by a quantum model with arbitrary precision.

Note that the statement that there exists some state $|\Gamma\rangle$ and some observable $M$ is equivalent to the statement that the target function can be learned by the relevant model (under the assumption the trainable circuit blocks are sufficiently flexible). As any frequency spectrum is asymptotically accessible, due to the assumption of a universal Hamiltonian family, the universality theorem is essentially equivalent to the statement that with sufficiently flexible circuit blocks such quantum models can realise any set of Fourier \textit{coefficients}.

\section{Practical implications for quantum machine learning}\label{Sec:implications}

In this last section, we discuss the scope and practical relevance of our results for quantum machine learning. First, we motivate that many quantum models proposed in the literature that do not immediately fit the base model from Eq. \eqref{eq:model} can still be analysed within our framework under the assumption that they encode classically pre-processed features $\boldsymbol \phi(\boldsymbol{x})$ instead of the original features $\boldsymbol{x}$. Second, we summarise guidelines that can help with the design of quantum machine learning algorithms.

\subsection{Classical pre-processing} \label{Sec:hybrid}

The base model used in this paper makes the assumption that a data feature is encoded into a subroutine $S(x)$ which consists of gates $\mathcal{G}(x) = e^{-ix H}$. We motivate in this section that many quantum machine learning algorithms which use other strategies of data encoding actually perform an implicit pre-processing of the data, and then use the ``time-evolution'' encoding studied here. The results of this paper are hence valid for the new features resulting from the pre-processing step. 

For example, the standard encoding procedure of traditional~\cite{nielsen2002quantum} (and some NISQ~\cite{farhi2018classification}) quantum algorithms, associates the $n$-bit binary representation of each (scalar) input feature $x$ with an $n$-qubit basis state, such as $x \mapsto \ket{0 1 0 1 1}$. The pre-processing step therefore maps original features $x$ to the angles $\boldsymbol \phi (x) = (\phi_1(x), \dots, \phi_n(x))$ with which the $n$ qubits have to be rotated to reflect every binary decimal digit of $x$ (i.e., $\pi$ for $\ket{1}$ or $0$ for $\ket{0}$). Our investigation here states that the quantum model for a single input feature corresponds to a multi-dimensional Fourier series in the angles, with a frequency spectrum size of at most $n$. In other words, the pre-processing changed the accessible Fourier spectrum by changing the features.

Another example is so-called ``amplitude encoding'' (i.e., \cite{rebentrost2014quantum, schuld2020circuit}), which associates an input vector $\bx$ with the values of the amplitudes of a quantum state. Practically, this requires $S(\bx)$ to be an arbitrary state preparation routine that is parametrised by some angles computed from $\bx$. The classical pre-processing therefore maps the original input to the set of angles used in the state preparation, $\bx \mapsto \boldsymbol \phi (\bx)$.

Pre-processing is also sometimes used in encoding strategies that directly feed input features into Pauli rotations. One example was used in Figs.~\ref{fig:fitting} and \ref{fig:L1-5}, where we re-scaled the inputs by a classical hyperparameter. In Ref.~\cite{perez2020data} it has been proposed to make these hyperparameters trainable (which in the light of the present analysis would allow for an adaptive ``frequency matching'' and may help to increase the expressivity of small quantum circuits). Another example is to construct higher-order features that are arithmetic combinations of the original inputs, like $\phi_1(\boldsymbol{x}) = x_1x_2, \phi_2(\boldsymbol{x}) = x_2x_3, \dots$, as used in the quantum feature map proposed in Ref.~\cite{havlivcek2019supervised}.

These examples suggest that implicit pre-processing can extend the function classes that quantum models can learn even further. However, care needs to be taken when making theoretical claims about the power of a quantum machine learning algorithm, which is, strictly speaking, a result of the quantum algorithm \textit{plus} the specific pre-processing strategy. In particular, comparisons to classical machine learning models should identify the pre-processing strategy and consider feeding the same pre-processed features to the classical model.

\subsection{Practical insights}

Finally, we want to summarise how the results of this paper can be used to understand and evaluate different design decisions of quantum machine learning models:

\emph{1.}\hspace{1em}If data is encoded via a Hamiltonian time evolution, we can naturally describe the class of functions that quantum models can learn as partial Fourier series. The Hamiltonian defines the available frequencies in the series, and the gates that do not encode data define the Fourier coefficients.

\emph{2.}\hspace{1em}If data is encoded into single-qubit Pauli rotations, the number of rotations used limits the number of frequencies that the model has access to. Repeating an encoding gate can help to increase the frequency spectrum, and thereby the expressivity of a quantum model.

\emph{3.}\hspace{1em}Quantum models naturally learn periodic functions in the data. One should therefore consider appropriate data re-scaling strategies, to make sure the data lies within the period of the function class. The natural representation of quantum models as Fourier series may suggest that time-series learning and signal processing tasks are particularly suitable applications for quantum machine learning. It may also hint at inherent regularising properties of quantum models that exclude higher-order Fourier frequencies.

\emph{4.}\hspace{1em}Classical pre-processing of the data, such as creating more features, can give small models more expressivity by enriching the frequency spectrum. 

\emph{5.}\hspace{1em}Adjusting the entries of the observable freely was a key ingredient in proving universality of quantum circuits in Section~\ref{Sec:universality}. Fixing the observable in a quantum model therefore limits its applicability. This fact suggests that parametrising the observable itself may be a key ingredient for flexible quantum models.

Ideally, one would hope that our results could provide concrete guidelines for the design of quantum machine learning models. However, in practical settings the process of model selection should be guided not purely by model expressivity, but rather through the expected generalisation performance of the model function class, as captured by capacity metrics such as the VC-dimension or Rademacher complexity \cite{shalev2014understanding}. While such capacities can be calculated for very simple function classes, calculating such metrics for more complex model classes, such as the quantum models studied here, is significantly harder. Additionally, in modern over-parametrised  models, which can often fit even randomised training data perfectly \cite{zhang2016understanding}, more sophisticated approaches are necessary to understand generalisation capacity \cite{jiang2019fantastic}. In light of this, the insights on how to make models more expressive should not be misinterpreted as recommendations for how to design \textit{good} quantum models -- a question which is much more complex and whose answer depends strongly on the context.

\section{Conclusion}

In this work we presented a systematic mapping between a large class of quantum machine learning models and partial Fourier series, which has allowed us to explore and quantify the effect of commonly used data-encoding mechanisms on the expressivity of these quantum models. We believe that this framework both lays a foundation for further theoretical analysis, and can serve as a useful guide in the search for suitable applications of such models. Additionally, this work provides a connection between quantum machine learning and ideas from the classical machine learning literature, such as neural networks with periodic activation functions \cite{sitzmann2020implicit}, and parametrised Fourier series as an alternative to neural networks \cite{zhumekenov2019fourier, wahls2014learning}. 

As mentioned throughout the paper, a variety of interesting questions remain. Firstly, can the framework developed here help us to understand and quantify the generalisation capacity of quantum models, and therefore guide model selection in a meaningful way? In particular, by using the representation of a quantum model as a partial Fourier series, can one calculate meaningful modern generalisation measures \cite{jiang2019fantastic} and use these for the development of model-selection guidelines? Secondly, we have proven our universality result under the assumption of exponential depth trainable circuit blocks (which provides a reasonable notion of \textit{asymptotic} universality with respect to circuit depth). In practical settings however one is interested in trainable circuit blocks with depth restrictions. Can one prove universality of such quantum models with either constant, logarithmic or polynomial depth trainable circuit blocks? In order to answer this question our toolbox needs to be developed further to understand how the structure of the trainable circuit blocks influences the set of accessible Fourier coefficients. Finally, it is currently unclear for which concrete applications quantum models may be naturally suited, or offer any sort of advantage over classical techniques, such as neural networks. Another question is therefore whether one can use knowledge of the function class expressed by quantum models, as developed in this work, to suggest natural applications for quantum machine learning.

\section*{Code} Code to reproduce the figures and explore further settings can be found in the following GitHub repository: \url{https://github.com/XanaduAI/expressive_power_of_quantum_models}.

\section*{Acknowledgements} MS wants to thank Nathan Killoran, Nicolas Quesada and Josh Izaac for helpful discussions. RS and JJM acknowledge funding from the BMWi under the PlanQK initiative. The authors endorse Scientific CO\textsubscript{2}nduct \cite{conduct} and provide a CO\textsubscript{2} emission table in Appendix~\ref{sec:appconduct}.

\bibliography{lit}

\clearpage
\appendix
\onecolumngrid

\section{Partial Fourier Series Representation of Multivariate Functions}\label{app:multivariate_FS}
In this section we show how a certain class of $L=1$ quantum models naturally realise multivariate Fourier series. On the one hand, this shows a way in which the univariate case analysed in the paper can easily generalise to multivariate models by encoding the features into different quantum subsystems. On the other hand, the multivariate model described in this section is a quantum model whose asymptotic universality is stated and discussed in Section \ref{Sec:universality}, and proven in Appendix~\ref{app:thm_proof}. 

More specifically, we consider a quantum model of the form
\begin{align}
    f(\vec{x}) = \langle 0 | \left( (W^{(1)})^{\dagger}S^{\dagger}(\vec{x}) (W^{(2)})^{\dagger}\right)M\left(W^{(2)}S(\vec{x}) W^{(1)}\right) |0\rangle,
\end{align}
where 
\begin{equation}
    S(\vec{x}) := e^{-i x_1 H_1}\otimes\ldots\otimes e^{-i x_N H_N}.
\end{equation}
Without loss of generality, instead of explicitly considering arbitrary unitaries $W^{(1)}$ and $W^{(2)}$, we can ``absorb" the unitaries into the initial state and measurement and consider the equivalent model
\begin{align}
    f(\vec{x}) = \langle \Gamma | S^{\dagger}(\vec{x}) MS(\vec{x}) |\Gamma\rangle,
\end{align}
where 
\begin{align}
    |\Gamma\rangle &= \sum_{j_1,\ldots,j_N = 1}^{2^d}\gamma_{j_1,\ldots,j_N}^{\vphantom{*}}|j_1\rangle\otimes\ldots\otimes|j_N\rangle
\end{align}
is some arbitrary state, and $M$ is some arbitrary observable. To simplify the index handling, we introduce the multi-indices $\jj \in [2^d]^N$ with which we can rewrite 
\begin{align}
    |\Gamma\rangle :=\sum_{\vec{j}}^{\vphantom{*}}\gamma_{\vec{j}}|\vec{j}\rangle.
\end{align}
Additionally, as argued before we can without loss of generality assume that all Hamiltonians are diagonal, i.e., that
\begin{equation}
    H_k = \mathrm{diag}(\lambda^{(k)}_1,\ldots, \lambda^{(k)}_{2^d}).
\end{equation}
With this assumption, we note that $S(\vec{x})$ is diagonal with entries
\begin{equation}
    [S(\vec{x})]_{\vec{j},\vec{j}} = e^{-i\vec{x}\cdot \vec{\lambda}_{\vec{j}}},
\end{equation}
where we have defined
\begin{equation}
    \vec{\lambda}_{\vec{j}} = (\lambda^{(1)}_{j_1},\ldots,\lambda^{(N)}_{j_N}).
\end{equation}
Given this, we see that
\begin{align}
        f(\vec{x}) &= \sum_{\vec{j}}\sum_{\vec{k}} \gamma^*_{\vec{j}}\gamma_{\vec{k}}^{\vphantom{*}}[S^{\dagger}(\vec{x})MS(\vec{x})]_{\vec{j},\vec{k}}^{\vphantom{*}} \\
        &= \sum_{\vec{j}}\sum_{\vec{k}} \gamma^*_{\vec{j}}\gamma_{\vec{k}}^{\vphantom{*}}M_{\vec{j},\vec{k}}^{\vphantom{*}}e^{i\vec{x}\cdot(\vec{\lambda}_{\vec{k}} - \vec{\lambda}_{\vec{j}})},
        \label{eq:multivariate_FS}
\end{align}
which is indeed a partial multivariate Fourier series, with the accessible frequencies fully determined by the spectra of the encoding Hamiltonians $\{ H_k \}$, and the Fourier coefficients determined by the trainable unitaries (or equivalently, the state and observable).

\section{Non-integer frequencies}\label{app:non_integer_frequencies}
In the main text, we put our focus on quantum models with integer-valued frequency spectra, as they naturally arise when using Pauli rotation gates and allow for analysis with the techniques of Fourier series. Here we will briefly discuss why many quantum models with non-integer-valued frequency spectra can be treated similarly. 

First, note that we can always decompose functions of the form $e^{i \omega x}$ into a Fourier series of integer-valued frequencies, i.e.,
\begin{align}
    e^{i \omega x} = \sum_{n=-\infty}^{\infty} \frac{(-1)^n \sin \omega \pi}{(\omega - n) \pi} e^{inx}
    = \sum_{n=-\infty}^{\infty} \operatorname{sinc}( \omega -n )e^{inx},
\end{align}
with $\operatorname{sinc}(z) = \sin(\pi z) / \pi z$. However, as we can see from this expression, any non-integer frequency in general ``contributes'' to \emph{infinitely many} Fourier coefficients. It turns out that a rather general case of quantum models with non-integer frequencies can be handled equivalently, namely if the frequencies are integer multiples of some basic frequency $\omega_0$,
\begin{align}
    \Omega = \{0, \pm n_1 \omega_0, \pm n_2 \omega_0, \dots \}.
\end{align}
This condition is equivalent to all frequencies in $\Omega$ being mutually commensurable, i.e., the ratio of any two frequencies is a rational number. This is the case in many natural settings, for example if encodings with non-integer frequencies are repeated in parallel alike to the Pauli encodings in Section~\ref{Sec:expressivity}. 

Basis functions of the form $e^{i x n \omega_0}$ are periodic functions on the interval $[0, 2\pi/\omega_0]$. This means that the generated Fourier-type sum in Eq.~\eqref{eq:fourier_sum} can be understood like the partial Fourier series in Eq.~\eqref{eq:partial_fourier_series}, but on a different interval. Alternatively, one can imagine re-scaling the data by $\tilde{x} =  x/\omega_0$, with which
\begin{align}
    e^{i \frac{x}{\omega_0} \omega } = e^{i \tilde{x} n}.
\end{align}

While this strategy could in principle be applied to any frequency spectrum where the frequencies are mutually commensurable, one has to be aware that $\omega_0$ is as least as small as the smallest difference between frequencies in $\Omega$. If very close frequencies are present in the spectrum, the data will have to be re-scaled by a very large factor to an interval where the generated Fourier coefficients may be sparse and the approximation quality poor.

\section{Proof of the universality theorem}\label{app:thm_proof}
We provide in this section a proof of the universality theorem stated in Section \ref{Sec:universality}, which we restate for completeness:
\begin{theorem*}\label{thm:universality_app}
Let $\{H_m\}$ be a universal Hamiltonian family, and $\{f_m\}$ the associated quantum model family, defined via Eq.~\eqref{eq:indexed_model}.  For all functions $g \in L_2([0,2\pi]^N)$, and for all $\epsilon > 0$, there exists some $m'\in \mathbb{N}$, some state $|\Gamma\rangle \in \mathbb{C}^{d^{m'}}$, and some observable $M$ such that
\begin{equation}
    \lvert\rvert f_{m'} - g \lvert\rvert_2 \leq \epsilon.
\end{equation}
\end{theorem*}
\begin{proof}
To begin with, we note that we can approximate any given $g \in L_2([0,2\pi]^N)$, up to an arbitrarily small error in $L_2$ norm, by using a truncated Fourier series~\cite{weisz2012summability}. More specifically, for any given $\epsilon > 0$, there exists some $K \in \mathbb{N}$ and some set of coefficients $\{c_{\vec{n}}\,|\, \vec{n}\in \mathbb{Z}^N_{K}\}$, with $c_{\vec{n}} = c^*_{-\vec{n}}$, such that
\begin{align}
    \tilde{g}(\vec{x}) &= \sum_{n_1 = -K}^{K}\ldots \sum_{n_N = -K}^{K} c_{\vec{n}}e^{i\vec{x}\cdot\vec{n}} \\&:= \sum_{\vec{n} \in \mathbb{Z}^N_{K} }c_{\vec{n}}e^{i\vec{x}\cdot\vec{n}}
\end{align}
satisfies
\begin{equation}
   \lvert\rvert \tilde{g} - g\lvert\rvert \leq \epsilon.
\end{equation}

In order to prove the theorem we therefore only need to show that there exists an $m' \in \bbN$, some state $|\Gamma\rangle$ and some observable $M$ so that the associated quantum model $f_{m'}$ generates the Fourier series $\tilde{g}$. 
Recall that the quantum model was defined as
\begin{align}
    f_m(\vec{x}) = \langle \Gamma | S^{\dagger}_{H_m}(\vec{x}) MS_{H_m}(\vec{x}) |\Gamma\rangle,
\end{align}
with 
\begin{equation}
    S_{H_m}(\vec{x}) := e^{-i x_1 H_m}\otimes\ldots\otimes e^{-i x_N H_m}.
\end{equation}
In Appendix~\ref{app:multivariate_FS}, we have seen that we can express the output of the model as
\begin{align}\label{eqn:model_expression}
    f_m(\vec{x}) &= \sum_{\vec{j}}\sum_{\vec{k}} \gamma^*_{\vec{j}}\gamma_{\vec{k}}^{\vphantom{*}}M_{\vec{j},\vec{k}}^{\vphantom{*}}e^{i\vec{x}\cdot(\vec{\lambda}_{\vec{k}} - \vec{\lambda}_{\vec{j}})},
\end{align}
where the multi-indices $\jj$ and $\vec{k}$ have $N$ entries that iterate over all $2^d$ basis states of the $d$ qubit subsystems. Let $\Omega_{H_m}$ be the frequency spectrum of $H_m$, as defined in Eq. \eqref{eq:freq_spec_ham}. As the $\{H_m\}$ form a universal family of Hamiltonians by assumption, we can choose an $m' \in \bbN$ so that 
\begin{align}
    \bbZ_K = \{-K, \dots, 0, \dots, K \} \subseteq \Omega_{H_{m'}}.
\end{align}
The accessible frequency vectors $\llambda_{\jj} - \llambda_{\vec{k}}$ independently contain all possible combinations of the frequencies in $\Omega_{H_{m'}}$. The vector-valued frequency spectrum for the multivariate case is therefore the Cartesian product of $N$ copies of $\Omega_{H_{m'}}$:
\begin{align}
    \Omega = \underbrace{\Omega_{H_{m'}} \times \dots \times \Omega_{H_{m'}}}_{N\text{ times}}.
\end{align}
As $\bbZ_K \subseteq \Omega_{H_{m'}}$ we naturally have that $\bbZ_K^N \subseteq \Omega$, which means that the Fourier series generated by the chosen model contains all terms that are necessary to construct the Fourier series $\tilde{g}$.

We can now revisit Eq.~\eqref{eqn:model_expression} and show that we can leverage the freedom of choosing both the initial state $\ket{\Gamma}$ and the observable $M$ arbitrarily to adjust all terms in the sum of Eq.~\eqref{eqn:model_expression} freely up to the complex-conjugation symmetry that guarantees that the model output is a real-valued function. To this end, we first observe that an exchange of the multi-indices $\jj$ and $\vec{k}$ yields the complex conjugate of the original term:
\begin{align}
    \left[\gamma^*_{\vec{j}}\gamma_{\vec{k}}^{\vphantom{*}} M_{\vec{j},\vec{k}}^{\vphantom{*}} e^{i\vec{x}\cdot(\vec{\lambda}_{\vec{k}} - \vec{\lambda}_{\vec{j}})}\right]^{*} &=
    \gamma^*_{\vec{k}}\gamma_{\vec{j}}^{\vphantom{*}} M_{\vec{j},\vec{k}}^{*} e^{i\vec{x}\cdot(\vec{\lambda}_{\vec{j}} - \vec{\lambda}_{\vec{k}})} \\
    &=
    \gamma^*_{\vec{k}}\gamma_{\vec{j}}^{\vphantom{*}} M_{\vec{k},\vec{j}}^{\vphantom{*}} e^{-i\vec{x}\cdot(\vec{\lambda}_{\vec{k}} - \vec{\lambda}_{\vec{j}})}.
\end{align}
Other than that, the coefficients can be freely chosen. To this end, we fix our initial state as the equal superposition state which can be prepared by applying a Hadamard gate to every qubit in the system. This gives $\gamma_{\jj} = 1/\sqrt{2^{Nd}}$ and results in the model
\begin{align}
    f_{m'}(\vec{x}) &= 2^{-Nd} \sum_{\vec{j}}\sum_{\vec{k}} M_{\vec{j},\vec{k}}^{\vphantom{*}} e^{i\vec{x}\cdot(\vec{\lambda}_{\vec{k}} - \vec{\lambda}_{\vec{j}})}.
\end{align}
With this choice, we see that the coefficients are directly proportional to the different entries of the observable $M$. Recall that our initial goal was to construct the Fourier series $\tilde{g}$ with coefficients $\{ c_{\nn}\}$ where $\nn \in \bbZ_K^N$. We already argued that all those are accessible in the frequency spectrum of our model because of the universal nature of the Hamiltonian family $\{ H_m \}$. As any frequency corresponds to one or more pairings of multi-indices $\jj$ and $\vec{k}$, we can always select a set of these multi-indices such that it is in one to one correspondence with the frequencies present in the Fourier series $\tilde{g}$:
\begin{align}
    I = \{ (\jj, \vec{k}) \in [2^d]^N \times [2^d]^N \, | \, \text{for all } \nn \in \bbZ_K^N \text{ there is exactly one pair } (\jj, \vec{k}) \text{ so that } \llambda_{\jj} - \llambda_{\vec{k}} = \nn \}.
\end{align}
With this it is now straightforward to use the freedom to choose our observable to fix $f_{m'} = \tilde{g}$ by choosing the diagonal and upper-triangular elements of $M$ via
\begin{align}
    M_{\jj, \vec{k}}^{\vphantom{*}} = \begin{cases}
        2^{Nd} c_{\nn} \text{ if } \llambda_{\jj} - \llambda_{\vec{k}} = \nn \text{ and } (\jj, \vec{k}) \in I \\
        0 \text{ otherwise}
    \end{cases}
\end{align}
after which the lower-triangular elements are fixed by the constraint that the observable is Hermitian.
\end{proof}

\newpage
\section{CO\textsubscript{2} Emission Table}\label{sec:appconduct}
\begin{table}[h]
\label{tab:cotwo}
\begin{tabular}[b]{l c}
\toprule
\textbf{Numerical simulations} & \\
\midrule
Total Kernel Hours [$\mathrm{h}$]& $\approx 100$ \\
Thermal Design Power Per Kernel [$\mathrm{W}$]& $\approx 50$\\
Total Energy Consumption Simulations [$\mathrm{kWh}$] & $\approx 5$ \\
Average Emission Of CO$_2$ In South Africa [$\mathrm{kg/kWh}$]& $\approx 1.5$\\
Total CO$_2$ Emission For Numerical Simulations [$\mathrm{kg}$] & $\approx 7.5$ \\
\midrule
\textbf{Transport} & \\
\midrule
Total CO$_2$ Emission For Transport [$\mathrm{kg}$] & 0\\
\midrule
Total CO$_2$ Emission [$\mathrm{kg}$] & $\approx 7.5$ \\
Were The Emissions Offset? & Yes \\
\bottomrule
\end{tabular}
\end{table}

\end{document}

%% file: commands.tex
\providecommand{\myvec}[1]{\ensuremath{\boldsymbol{#1}}}

\providecommand{\jj}{\ensuremath{\myvec{j}}}
\providecommand{\kk}{\ensuremath{\myvec{k}}}

\providecommand{\nn}{\ensuremath{\myvec{n}}}

\providecommand{\llambda}{\ensuremath{\myvec{\lambda}}}

\providecommand{\oomega}{\ensuremath{\myvec{\omega}}}

\providecommand{\bbN}{\ensuremath{\mathbb{N}}}

\providecommand{\bbZ}{\ensuremath{\mathbb{Z}}}